\documentclass[manuscript,nonacm]{acmart}

\usepackage{algorithm,algorithmic}
\usepackage{multicol}

\usepackage{stackengine,scalerel}
\usepackage{calc}

\usepackage{pgfplots}  
\pgfplotsset{width=10cm,compat=1.9}

\usepackage{float,wrapfig}
\setcopyright{none}
\copyrightyear{2021}
\acmYear{2021}
\acmDOI{}
\acmConference[]{}{}{}

\acmPrice{}
\acmISBN{}

\usepackage{xcolor}

\newcommand{\remove}[1]{}

\begin{document}
\pagestyle{plain}

\title{Blunting an Adversary Against Randomized Concurrent Programs with Linearizable Implementations}


\author{Hagit Attiya}
\email{hagit@cs.technion.ac.il}
\affiliation{%
  \institution{Technion}
  \city{Haifa}
  \country{Israel, {\sf hagit@cs.technion.ac.il}}
}

\author{Constantin Enea}
\email{cenea@irif.fr}
\affiliation{%
  \institution{LIX, Ecole Polytechnique, CNRS and Institut Polytechnique de Paris}
  \city{Palaiseau}
  \country{France, {\sf cenea@irif.fr}}
}

\author{Jennifer L.~Welch}
\email{welch@cse.tamu.edu}
\affiliation{%
  \institution{Texas A\&M University}
  \city{College Station, TX}
  \country{USA, {\sf welch@cse.tamu.edu}}
}

\begin{abstract}
Atomic shared objects, whose operations take place instantaneously,
are a powerful abstraction for designing complex concurrent programs.
Since they are not always available,
they are typically substituted with software implementations. 
A prominent condition relating these implementations 
to their atomic
specifications is {\em linearizability}, which preserves safety
properties of the programs using them.
However linearizability does not preserve \emph{hyper-properties},
which include probabilistic guarantees of randomized programs:
an adversary can greatly amplify the probability of a bad outcome,
such as nontermination, by manipulating the order of events inside
the implementations of the operations.
This unwelcome behavior prevents modular reasoning,
which is the key benefit provided by the use of linearizable object
implementations.
A more restrictive property, \emph{strong linearizability}, does
preserve hyper-properties but it is impossible to achieve in
many situations.

This paper suggests a novel approach to blunting the adversary's
additional power that works even in cases where
strong linearizability is not achievable.
\emph{We show that a wide class of linearizable implementations,
including well-known ones for
registers and snapshots, can be
modified to approximate the probabilistic guarantees of randomized
programs when using atomic objects.}
The technical approach is to
transform the algorithm of each method of an existing linearizable
implementation by repeating a carefully chosen prefix of the method several
times and then randomly choosing which repetition to use subsequently.
We prove that the probability of a bad outcome decreases with the
number of repetitions, approaching the probability attained when
using atomic objects.
The class of implementations to which our transformation applies
includes the ABD implementation of a shared register using
message-passing, the Afek et al.\
implementation of an atomic snapshot using single-writer registers,
the Vit\'{a}nyi and Awerbuch implementation of a multi-writer register
using single-writer registers,
and the Israeli and Li implementation of a
multi-reader register using single-reader registers,
all of which are widely used in asynchronous crash-prone systems.
\end{abstract}

\maketitle


\section{Introduction}
\label{section:intro}

Atomic shared objects, whose operations take place instantaneously,
are a powerful abstraction for designing complex concurrent programs,
as they allow developers to reason about their programs in terms of
familiar data structures.
Since they are not always available, they are typically
substituted with software implementations.
A prominent condition relating these implementations
to their atomic
specifications is {\em linearizability}~\cite{HerlihyW1990}.
It provides the illusion that processes
communicate through shared objects on which operations occur
instantaneously in a sequential order,
called the \emph{linearization order},
regardless of the actual communication mechanism.
A key benefit of linearizability is that it preserves any safety
property enjoyed by the
program when it is executed with atomic objects.

Unfortunately, linearizability does not preserve
\emph{hyper-properties}~\cite{ClarksonS2010},
which include probabilistic guarantees of randomized programs.
As demonstrated by examples
in~\cite{GolabHW2011,AttiyaE2019,HadzilacosHT2020randomized},
an adversary can greatly amplify the probability of a bad outcome,
such as nontermination, by manipulating the order of events inside
the implementations of the operations.
Such behavior invalidates the key benefit of using linearizable
objects, which is the modularity that they provide by hiding implementation
details behind an interface that mimics atomic behavior.
To overcome this drawback, Golab, Higham and Woelfel~\cite{GolabHW2011}
proposed a more restrictive property, {\em strong linearizability},
that preserves hyper-properties, including probability distributions.
However, not many strongly-linearizable implementations are
known and in fact they are impossible in several important cases
(cf.\ Section~\ref{section:related}).

This paper suggests a novel approach to blunting the adversary's
additional power that works even in cases where strong linearizability
is not achievable.
To motivate our approach, consider the well-known ABD~\cite{AttiyaBD1995}
linearizable implementation of a read-write register in crash-prone
message-passing systems and how it behaves
in the context of the simple program given in
Algorithm~\ref{algorithm:simplified-weakener}, which we distill
from the \emph{weakener} program~\cite{HadzilacosHT2020randomized}.
In the multi-writer version of ABD~\cite{LynchS1997}, which we
consider throughout, both read and write operations start with a
``query'' message-exchange phase in which the invoking process obtains
the timestamp associated with the most recent value.
Then, both operations execute an ``update'' message-exchange phase;
the reader announces the latest value and timestamp before returning
the value, while the writer announces the new value and assigns it a larger
timestamp.
The linearization order of the operations is completely determined by
the maximal timestamps that are obtained during the query phases,
and thus, their order is determined at the end of the query phase.

\begin{wrapfigure}{r}{0.5\textwidth}
\begin{minipage}{0.5\textwidth}
\begin{algorithm}[H]
\caption{
  \small Processes $p_0$, $p_1$, and $p_2$ share two registers $R$,
  written by $p_0$ and $p_1$ and read by $p_2$,
  and $C$, written by $p_1$ and read by $p_2$.}
\label{algorithm:simplified-weakener}
\begin{algorithmic}[1]
\STATE Initially:  $R = \bot$, $C = -1$
\STATE {\bf Code for $p_i$, $i \in \{0,1\}$:}
\STATE  $R$ := $i$
\STATE  if ($i = 1$) then $C$ := flip fair coin (0 or 1)
  \label{line:sim-weak-flip}
\STATE {\bf Code for $p_2$:}
\STATE $u_1$ := $R$; $u_2$ := $R$; $c$ := $C$
\STATE if (($u_1 = c$) $\wedge$ ($u_2 = 1 - c$)) then loop forever
   \label{testline}
\STATE else terminate
\end{algorithmic}
\end{algorithm}
\end{minipage}
\vspace{1mm}
\end{wrapfigure}

Algorithm~\ref{algorithm:simplified-weakener}
has two processes, $p_0$ and $p_1$, that write their ids to register $R$,
then $p_1$ flips a coin and writes the result to another register $C$.
A third process, $p_2$, reads $R$ twice and $C$ once; if it succeeds
in reading both ids from $R$ and the first id that it reads equals
the result of the coin flip, then it loops forever,
otherwise it terminates.
When the registers are atomic, $p_2$ terminates with probability
at least one-half, for any adversary.
(See Appendix~\ref{ssec:weakener:atomic} for details.)
Yet when the registers are replaced with ABD implementations, a strong
adversary, which can observe processes' random choices~\cite{Aspnes2003},\footnote{
    Throughout this paper we consider only strong adversaries
    and sometimes drop the term ``strong''.}
can interleave the internal steps of the query phase
and the steps of the program so as to ensure that $p_2$
\emph{never} terminates.
(See Appendix~\ref{ssec:app_ABD} for details.)

Instead of attempting to find a strongly-linearizable replacement
for ABD, which is impossible~\cite{AttiyaEW2021,ChanHHT2021arxiv},
we make the key observation that the adversary can disrupt the workings
of the program only when the coin flip on Line~\ref{line:sim-weak-flip}
occurs during the query phase of a read or write operation.
The reason is that, after the query phase has completed, the linearization
order of the operation is fixed.
We also observe that
the query phase is ``effect-free'' in the sense that
it can be repeated multiple times without the repetitions interfering
with each other or with the behavior of the other processes.

Our modification to ABD is for each operation to \emph{execute the
query phase several times, and then randomly choose which one of the
values obtained} to use in the rest of the operation.
In Algorithm~\ref{algorithm:simplified-weakener},
the adversary can make only \emph{one} of these values depend
on the result of the coin flip (by scheduling the coin flip during
that iteration of the query phase), but
that value is used in the rest of the operation
with some probability strictly smaller than 1,
since values from query phases are chosen uniformly at random.
As a result, the program exhibits probabilistic behavior closer
to that seen with atomic objects.
For example, repeating the query phase twice when ABD
is used in Algorithm~\ref{algorithm:simplified-weakener}
ensures that $p_2$ terminates with probability at least 1/8,
in contrast with the 0 termination probability when using
the original ABD implementation.
(See Appendix~\ref{subsec:abd2-prob} for details.)
Thus by carefully introducing additional randomization inside the
linearizable implementation itself, we blunt the power of the adversary
to disrupt the behavior of the randomized program using the object,
while keeping the implementation linearizable.

We generalize this idea to develop a transformation for the class of
linearizable implementations
in which operations can be partitioned, informally speaking,
into an \emph{effect-free preamble} followed by a tail in which the
operation's linearization order is fixed.
The latter property is made precise under the notion of
\emph{tail strong linearizability} (Section~\ref{sec:tls}).
Our \emph{preamble-iterating transformation}
(Section~\ref{ssec:transf_tls}) repeats
the preamble in the implementation of each operation some number of times
and then randomly chooses the results of one repetition to use, producing
a linearizable implementation of the same object.

\emph{Our main result (Theorem~\ref{th:main})
is that the probability of the program reaching a bad
outcome with the transformed objects approaches the probability of reaching
the same bad outcome with atomic versions of the objects, as the number of
repetitions of the preamble increases relative to the number of
random choices made in the program}.
Specifically, we show that the probability of the bad outcome using
the transformed object is at most the probability of the bad outcome using
atomic objects, which is the best case, plus a fraction of the difference
between the probabilities of the bad outcome when using the linearizable
objects and using the atomic objects.
The fraction is the probability that adversary is able to manipulate
the behavior to its advantage, and it decreases as the number of repetitions
increases.

Our transformation applies to a broad class of both
shared-memory and message-passing implementations that are widely used,
and includes ABD (both its original single-writer
version~\cite{AttiyaBD1995} and its multi-writer
version~\cite{LynchS1997}),
the atomic snapshot algorithm~\cite{AfekADGMS1993},
the algorithm to construct a multi-writer register using
single-writer registers~\cite{VitanyiA1986}, and
the algorithm to construct a multi-reader register
using single-reader registers~\cite{IsraeliL1993}.
To summarize our contributions:
\begin{itemize}
\item We introduce a new strengthening of linearizability called
  \emph{tail strong linearizability} which, roughly speaking,
  imposes the requirements of strong linearizability only on executions
  in which each operation has passed its \emph{preamble}.
  (See Section~\ref{sec:tls} for the precise definition.)
  We show that this property is satisfied by a wide range of objects
  that also have effect-free preambles (Section~\ref{app:tasl_objects}).
\item We define a transformation of tail-strongly-linearizable objects
  with effect-free preambles, which
  iterates the preamble of each operation multiple times and
  then randomly chooses an iteration whose results will be used in the
  rest of the operation (Section~\ref{ssec:transf_tls}).
\item We characterize the blunting power of the ``preamble iterated''
  objects with a quantitative upper bound on the
  amount by which the probability of reaching a bad outcome increases
  when using the transformed objects instead of the atomic objects,
  and relative to using the original linearizable objects
  (Theorem~\ref{th:main} in Section~\ref{ssec:quantify}).
\end{itemize}


\newcommand{\impls}{\mathcal{I}}
\newcommand{\outs}{\mathcal{B}}
\newcommand{\objs}{\mathcal{O}}

\newcommand{\values}{\mathbb{V}}
\newcommand{\random}[1]{\ensuremath{\mathsf{random}({#1})}}
\newcommand{\grandom}{\random{V}}
\newcommand{\rd}[1]{\ensuremath{\mathsf{Read}({#1})}}
\renewcommand{\wr}[1]{{\mathsf{Write}({#1})}}
\newcommand{\exec}[3]{{e[{#1},{#2},{#3}]}}
\newcommand{\execs}{E}
\newcommand{\prob}[2]{{\mathit{OutDist}({#1},{#2})}}
\newcommand{\maxprob}[1]{{\mathit{Prob}_{{#1}}}}
\newcommand{\probOf}[1]{{\mathit{Prob}[{#1}]}}

\newcommand{\procs}{\mathbb{P}}
\newcommand{\states}{\mathbb{Q}}
\newcommand{\msgs}{\mathbb{M}\mathit{sgs}}
\newcommand{\meths}{\mathbb{M}\mathit{eths}}
\newcommand{\Objects}{\mathbb{O}\mathit{bjs}}
\newcommand{\vals}{\mathbb{V}}
\newcommand{\tup}[1]{\langle{#1}\rangle}
\newcommand{\calls}{\mathbb{C}a}
\newcommand{\returns}{\mathbb{R}a}
\newcommand{\acts}{\mathbb{A}}
\newcommand{\call}[1]{\mathit{call}\ {#1}}
\newcommand{\ret}[1]{\mathit{ret}\ {#1}}
\newcommand{\confs}{\mathbb{C}}
\newcommand{\cp}{\Pi}

\section{Preliminaries}

Randomized programs consist of a number of processes that invoke
methods
of some set of shared objects,
perform local computation, or sample values uniformly at random
from a given set of values.
We are interested in reasoning about the probability that
a strong adversary~\cite{Aspnes2003}
can cause a program
to reach a certain set of program outcomes,
defined as sets of values returned by method
invocations i.e., operations.

\subsection{Objects}

An \emph{object} is defined by a set of method names
and an implementation that defines the behavior of each method. Methods can be invoked in parallel at different processes.
In message-passing implementations, processes communicate by sending and receiving messages,
while in shared-memory implementations,
they communicate by invoking methods of a set of shared objects
(e.g., some class of registers) that execute instantaneously (in a single indivisible step), called \emph{base} objects.
The pseudo-code we will use to define such implementations can be translated in a straightforward manner to executions seen as sequences of labeled transitions between global states that track the local states of all the participating processes, the states of the shared base objects or the set of messages in transit, depending on the communication model, and the control point of each method invocation in a process.
Certain transitions of an execution correspond to initiating a new method invocation, called \emph{call transitions}, or returning from an invocation, called \emph{return transitions}.  Such transitions are labeled by call and return actions, respectively.
A \emph{call action} $\call{M(x)}_i$ labels a transition corresponding to invoking a method $M$ with argument $x$; $i$ is an identifier of this invocation. A \emph{return action} $\ret{y}_i$ labels a transition corresponding to invocation $i$ returning value $y$. For simplicity, we assume that each method
has at most one parameter and at most one return value.
We assume that each label of a transition corresponding to a step of an invocation $i$ includes the invocation identifier $i$ and the control point (line number) $\ell$ of that step. In particular, call transitions include an initial control point $\ell_0$. Such a transition is called a \emph{step} of $i$ at $\ell$.

The set of executions of an object $O$ is denoted by $E(O)$. An execution of an object $O$ satisfies standard well-formedness conditions, e.g., each transition corresponding to returning from an invocation $i$ (labeled by $\ret{y}_i$ for some $y$) is preceded by a transition corresponding to invoking $i$ (labeled by $\call{M(x)}_i$, for some $M$ and $x$), and for every $i$ there is at most one transition labeled by a call action containing $i$, and at most one transition labeled by a return action containing $i$.

An object where every invocation returns immediately is called \emph{atomic}. Formally, we say that an object $O$ is \emph{atomic} when every transition labeled by $\call{M(x)}_i$, for some $M$ and $x$, in an execution (from $E(O)$) is immediately followed by a transition labeled by $\ret{y}_i$ for some $y$.

Correctness criteria like linearizability characterize sequences of call and return actions in an execution, called \emph{histories}. The history of an execution $e$, denoted by $\mathit{hist}(e)$, is defined as the projection of $e$ on the call and return actions labeling its transitions. The set of histories of all the executions of an object $O$ is denoted by $H(O)$. Call and return actions $\call{M(x)}_i$ and $\ret{y}_i$ are called \emph{matching} when they contain the same invocation identifier $i$. A call action is called \emph{unmatched} in a history $h$ when $h$ does not contain the matching return. A history $h$ is called \emph{sequential} if every call $\call{M(x)}_i$ is immediately followed by the matching return $\ret{y}_i$. Otherwise, it is called \emph{concurrent}. Note that every history of an atomic object is sequential.

\subsection{(Strong) Linearizability}

\emph{Linearizability}~\cite{HerlihyW1990} defines a relationship between
histories of an object and a given set of sequential histories,
called a \emph{sequential specification}.  The sequential specification can also be interpreted as an atomic object.
Therefore, given two histories $h_1$ and $h_2$, we use $h_1\sqsubseteq h_2$ to denote the fact that there exists a history $h_1'$ obtained from $h_1$ by appending return actions that correspond to some of the unmatched call actions in $h_1$ (completing some pending invocations)
and deleting the remaining
unmatched call actions in $h_1$ (removing some pending invocations), such that $h_2$ is a permutation of $h_1'$ that preserves the order between return and call actions, i.e.,
if a given return action occurs before a given call action in $h_1'$ then the same holds in $h_2$. We say that $h_2$ is a \emph{linearization} of $h_1$.
A history $h_1$ is called \emph{linearizable} w.r.t.~a sequential specification $\mathit{Seq}$ iff there exists a sequential history $h_2\in \mathit{Seq}$ such that $h_1\sqsubseteq h_2$. An execution $e$ is linearizable w.r.t. $\mathit{Seq}$ if $\mathit{hist}(e)$ is linearizable w.r.t. $\mathit{Seq}$. An object $O$ is linearizable w.r.t.~$\mathit{Seq}$ iff each history $h_1\in H(O)$ is linearizable w.r.t.~$\mathit{Seq}$.

Two objects $O_1$ and $O_2$ are called \emph{equivalent} when they are linearizable w.r.t. the same sequential specification $\mathit{Seq}$ and for every history $h\in \mathit{Seq}$, $H(O_1)$ contains a history linearizable w.r.t. $h$ iff $H(O_2)$ contains a history linearizable w.r.t. $h$.

\emph{Strong linearizability}~\cite{GolabHW2011} is a strengthening of linearizability that enables preservation of probability distributions in randomized programs using a certain object $O$ instead of an \emph{atomic} object equivalent to $O$. It also enables preservation of
more generic hyper-safety properties~\cite{AttiyaE2019}.
A set of executions $E\subseteq E(O)$ of an object $O$ is
called \emph{strongly linearizable} when
it admits linearizations that are consistent with linearizations of prefixes that belong to $E$ as well.
Formally, $E$ is strongly linearizable w.r.t.~a sequential specification 
$\mathit{Seq}$ iff there exists a function
$f:E\rightarrow \mathit{Seq}$ such that:
\begin{itemize}
\item for any execution $e\in E$, $\mathit{hist}(e)\sqsubseteq f(e)$, and
\item $f$ is prefix-preserving, i.e.,
for any two executions $e_1,e_2\in E$ such that
$e_1$ is a prefix of $e_2$, $f(e_1)$ is a prefix of $f(e_2)$.
\end{itemize}
An object is called \emph{strongly linearizable} when its entire set of executions $E(O)$ is strongly linearizable.

\subsection{Randomized Programs}

A \emph{program} $P(\objs)$ is composed of a number of processes that
invoke methods on a set of shared objects $\objs$.  Besides shared
object invocations, a process can also perform some local computation
(on some set of local variables), and use an
instruction $\grandom$, where $V$ is a subset of a domain of values
$\values$, to sample a value from $V$ uniformly at random.  This value
can be used, for instance, as an input to a method invocation.
The syntax used for local computation instructions is not important,
and we omit a precise formalization.

An \emph{execution} of a program $P(\objs)$
is an interleaving of steps taken by the processes it contains. 
A step can correspond to either 
\begin{itemize}
\item 
an interaction with a shared object in $\objs$, i.e., 
a method invocation, internal step of an object implementation, 
or returning from a method, or
\item 
a local computation in the program, e.g., an execution of $\grandom$, 
for some $V$.
\end{itemize}
As expected, the sequence of steps in an execution follows the control-flow in each process and the internal behavior of the shared objects in $\objs$ (whether they be implemented on top of a message-passing or shared-memory system).

The \emph{outcome} of a program execution is a mapping from shared object method invocations to the values they return in that execution. In order to relate outcomes in different executions of the same program $P(\objs)$, we assume that shared object method invocations in executions of $P(\objs)$ have unique identifiers that relate to the \emph{syntax} of $P(\objs)$. These identifiers can be defined, for instance, as a triple of a process id, the control point (line number) at which that invocation occurs, and the number of times this control point occurred in the past (in order to deal with looping constructs). Then, an outcome maps these identifiers to return values. An outcome of a program $P(\objs)$ is the outcome of an execution of $P(\objs)$.

Consider two sets of objects $\objs_1$ and $\objs_2$ for which there exists a bijection $\lambda$ that maps each object $O\in\objs_1$ to an \emph{equivalent} object $O'\in \objs_2$. Given a program $P(\objs_1)$, the program $P(\objs_2)$ is obtained by substituting every object $O\in\objs_1$ with the corresponding object $\lambda(O)\in \objs_2$.

\begin{proposition}
$P(\objs_1)$ and $P(\objs_2)$ have the same set of outcomes.
\end{proposition}

\subsection{Adversaries}

We say that a program execution \emph{observes} a sequence of random values $\vec{v}$ if the $i$-th occurrence of a step that samples a random value
(by
executing a $\grandom$ instruction) returns $\vec{v}[i]$, where $\vec{a}[i]$ is the $i$-th position in a vector $\vec{a}$.
A \emph{schedule} is a sequence of process ids. An execution \emph{follows} a schedule $\vec{s}$ when the $i$-th step of the execution is executed by the process $\vec{s}[i]$. In the following, we assume \emph{complete} schedules that make the program terminate.
We denote by $\exec{P(\objs)}{\vec{v}}{\vec{s}}$ the unique execution of a program $P(\objs)$ that observes $\vec{v}$ and follows $\vec{s}$.

For a program $P(\objs)$, a \emph{(strong) adversary} $A$ against $P(\objs)$ is a mapping from sequences of values in $\values$ to complete schedules. We assume that for every two sequences $\vec{v_1},\vec{v_2}\in\values^*$
 that have a common prefix of length $m$, the executions $\exec{P(\objs)}{\vec{v_1}}{A(\vec{v_1})}$ and $\exec{P(\objs)}{\vec{v_2}}{A(\vec{v_2})}$ are the same until the $(m+1)$-th occurrence of a step that samples a random value, or the end of the execution if no such steps remain.  This assumption captures the constraint that the scheduling decisions of a strong adversary do not depend on future randomized choices. A strong adversary $A$ defines a set of executions $\execs(A)$, each of which observes a sequence of values $\vec{v}$ and follows the schedule $A(\vec{v})$.

An adversary $A$ against $P(\objs)$ defines a probability distribution over program outcomes (of executions in $\execs(A)$), denoted by $\prob{P(\objs)}{A}$. Given a set of outcomes $\outs$, $\probOf{P(\objs)||A\rightarrow\outs}$ is the probability defined by $\prob{P(\objs)}{A}$ of an outcome being contained in $\outs$. The \emph{probability of $P(\objs)$ reaching $\outs$}, denoted by $\probOf{P(\objs)\rightarrow\outs}$, is defined as the maximal probability $\probOf{P(\objs)||A\rightarrow\outs}$ over all possible adversaries $A$.
In the context of our results, the set of outcomes $\outs$ is interpreted as some set of ``bad'' states, and the goal is to minimize the probability of a program reaching them.

The following result shows that a program using atomic objects minimizes the probability of reaching a set of outcomes, among programs where the atomic objects can be replaced with equivalent ones. This follows from the fact that an adversary can restrict itself to schedules where each method invocation is executed in isolation (a method can be called only when there is no other pending call), and the outcomes obtained in executions following such schedules can also be obtained with executions of atomic objects.
For a set of objects $\objs$, $\objs_a$ is the set of \emph{atomic} objects $O'$ that are equivalent to objects $O\in\objs$.

\begin{proposition}\label{prop:atomic_bound}
For any program $P(\objs)$ and set of outcomes $\outs$,
$\probOf{P(\objs)\rightarrow\outs} \geq \probOf{P(\objs_a)\rightarrow\outs}$.
\end{proposition}

Algorithm~\ref{algorithm:simplified-weakener} is an example of a program $P$ where $\probOf{P(\objs)\rightarrow\outs}$ is strictly greater than $\probOf{P(\objs_a)\rightarrow\outs}$ (see Appendix~\ref{section:overview}). In this case, $\objs$ consists of two instances of the ABD register, one for $R$ and one for $C$, and $\outs$ is the set of outcomes where the return values of $p_2$'s invocations satisfy $u_1=c$ and $u_2=1-c$. These values make $p_2$ not terminate.

The two probabilities in Proposition~\ref{prop:atomic_bound}
are equal when $\objs$ is a set of strongly linearizable objects:

\begin{theorem}[{\cite{GolabHW2011}}]\label{th:strong_lin}
For any program $P(\objs)$ using a set of strongly linearizable objects $\objs$, and set of outcomes $\outs$,
$\probOf{P(\objs)\rightarrow\outs} = \probOf{P(\objs_a)\rightarrow\outs}$.
\end{theorem}

\section{Tail Strong Linearizability}\label{sec:tls}

We define a generalization of strong linearizability, called \emph{tail strong linearizability}, which requires that executions be mapped to prefix-preserving linearizations only when each method invocation has executed a minimal number of steps called a \emph{preamble}. The relationship between linearizations of different executions where some invocation has \emph{not} executed its preamble fully is unconstrained.  When the preamble of every invocation is ``empty'' (i.e., it includes only the call transition), this becomes the standard notion of strong linearizability. When the preamble of every invocation is ``full'' (i.e., it includes all the steps of the invocation), this is equivalent to standard linearizability (since linearizability requires anyway that any invocation $i$ is linearized before any other invocation $i'$ that starts after $i$ returns).
Section~\ref{section:ABD} defines a preamble-iterating transformation of tail strongly linearizable objects that limits the increase in the probability of a bad outcome
when a program uses the transformed objects instead of equivalent atomic
objects.

Let $O$ be an object with a set of methods $\meths$. A \emph{preamble mapping} $\cp$ of $O$ is a mapping that associates each method $M\in \meths$ with a control point $\ell$ representing the last step of its preamble. We assume that every control-flow path of $M$ should pass through $\ell$ and that $\ell$ can be reached only once (it is not inside the body of a loop). The trivial preamble mapping that associates each method to the initial control point $\ell_0$ is denoted by $\cp_0$.  For instance, for the multi-writer version of ABD (listed in Algorithm~\ref{algorithm:ABD} and described in the introduction), we are interested in a preamble mapping that associates the \textbf{Read} and \textbf{Write} methods with the control points where the value with the largest timestamp received from responses to query messages is assigned (Lines~\ref{line:ReadControl} and~\ref{line:WriteControl},
respectively, in Algorithm~\ref{algorithm:ABD}).

Given an execution $e$ and a method invocation $i$, we say that $i$ \emph{passed} a control point $\ell$ when $e$ contains a step of $i$ at $\ell$.
An execution $e$ is \emph{complete} w.r.t. a preamble mapping $\cp$ if each invocation of a method $M$ in $e$ passed the control point $\cp(M)$. The set of executions of $O$ complete w.r.t. $\cp$ is denoted by $E(O,\cp)$.

An object $O$ is called \emph{tail strongly linearizable} w.r.t.~a preamble mapping $\cp$ and a sequential specification $\mathit{Seq}$ when it is linearizable w.r.t. $\mathit{Seq}$ and the set of executions $E(O,\cp)$ is strongly linearizable w.r.t.~$\mathit{Seq}$. Note that strong linearizability is equivalent to tail strong linearizability w.r.t.~$\cp_0$.

When reasoning about programs that use more than one object, we rely on the fact that tail strong linearizability is \emph{local} in the sense that it holds for the union of a set of objects that are each tail strongly linearizable.  Locality holds for tail strong linearizability as a straightforward consequence of the fact that standard strong linearizability is local~\cite{GolabHW2011}.

\begin{theorem}\label{th:local}
A set of histories $H$ of executions with multiple objects $O_1$,$\ldots$,$O_m$ is tail strongly linearizable w.r.t. some preamble mapping $\cp_1\cup\ldots\cup\cp_m$, where $\cp_j$ is a preamble mapping of $O_j$,
iff for all $j$,
$1\leq j\leq m$, the set $H_j=\{h | O_j: h\in H\}$, where $h | O_j$ is the projection of $h$ on call and return actions of $O_j$, is tail strongly linearizable w.r.t. $\cp_j$.
\end{theorem}


\newcommand{\objTr}[2]{{O_{#2}^{#1}}}

\section{Blunting an Adversary Against Tail Strongly Linearizable Objects}
\label{section:ABD}

We define a methodology for transforming tail strongly linearizable
objects whose preambles have a certain property we call ``effect-free''
into equivalent objects.
The use of the transformed objects can reduce the probability that a
program using the objects reaches a set of (bad) outcomes.
Intuitively, the transformed objects can blunt the power of any
adversary against a program using them and in the limit
restrict its power to what it has when the program uses atomic objects
(which is a lower bound by Proposition~\ref{prop:atomic_bound}).
As we show in Section~\ref{app:tasl_objects},
the class of objects to which the transformation applies includes
a broad set of widely-used objects,
including the ABD register (both its original single-writer
version~\cite{AttiyaBD1995} as well as the multi-writer
version~\cite{LynchS1997}),
the atomic snapshot algorithm using single-writer registers
of Afek et al.~\cite{AfekADGMS1993},
the Vit\'{a}nyi and Awerbuch algorithm to construct a multi-writer register
from single-writer registers~\cite{VitanyiA1986},
and the Israeli and Li algorithm to construct a multi-reader register
from single-reader registers~\cite{IsraeliL1993}.
None of these implementations is strongly linearizable
and in fact strongly-linearizable implementations are known to be
impossible in most of these cases (see Section~\ref{section:related}).

\subsection{The Preamble-Iterating Transformation for
Tail Strongly Linearizable Objects}\label{ssec:transf_tls}

\newlength\myindent
\setlength\myindent{1em}
\newcommand\bindent{%
  \begingroup
  \setlength{\itemindent}{\myindent}
  \addtolength{\algorithmicindent}{\myindent}
}
\newcommand\eindent{\endgroup}

The preamble-iterating transformation is defined in Algorithm~\ref{algorithm:O-star}.
For a given integer $k\geq 1$, object $O$, and preamble mapping $\cp$, we define an object
$\objTr{k}{\cp}$ (we may omit the preamble mapping $\cp$ from the notation when it is understood from the context) where each method $M$ is replaced with a method $M^k$ that iterates the preamble of $M$ (see the \textbf{for} loop
in Algorithm~\ref{algorithm:O-star}) $k$ times and uses the values of a randomly chosen iteration for the rest of the code.
To simplify the notations, we assume that the code of each preamble of a method $M$ (the code up to and including the control point $\cp(M)$) is encapsulated in a function called \textsc{preamble} that takes the same input as $M$ and returns the values of $M$'s local variables after executing that preamble. These values are stored in the array $\mathit{locals}$. The rest of the code, which uses the values in $\mathit{locals}$, is left unchanged. The results of the preamble iterations are stored in a two dimensional array $\overrightarrow{\mathit{locals}}$ where each row has the same size as $\mathit{locals}$. For the ABD register, the ABD$^k$ object is listed in Algorithm~\ref{algorithm:ABD-star} (Appendix~\ref{section:overview}).

This transformation leads to an equivalent object provided that the preamble contains only \emph{effect-free} computation, which does not affect the behavior of the other processes running concurrently (effect-free computation can affect the state of the process that executes it). For instance, the preamble of ABD's \textbf{Read} and \textbf{Write} methods consists in sending ``query'' messages to the other processes, waiting for replies, and computing the largest timestamp value from the replies (the queryPhase function in Algorithm~\ref{algorithm:ABD}). Sending a reply to a query message from another concurrently running process does not affect the behavior of the sender, as its local variables remain unchanged.

\begin{wrapfigure}{r}{0.55\textwidth}
\begin{minipage}{0.55\textwidth}
\begin{algorithm}[H]
\caption{\small
Transforming a tail strongly linearizable object $O$ to $O^k$,
$k\geq 1$.
Each method $M$ of $O$ is transformed to a method $M^k$ of $O^k$.}
\label{algorithm:O-star}
\begin{algorithmic}[0]
\STATE{method \textbf{M($v$):}}
\bindent
\STATE{ $\mathit{locals}$ := \textsc{preamble}({\em v})}
\STATE{ // rest of the code $\ldots$ }
\eindent

\STATE method {\bf M$^k$($v$):}
\bindent
\FOR{$i$ := 1 \TO $k$ }
\STATE{ $\overrightarrow{\mathit{locals}}[i]$ := \textsc{preamble}($v$)}
\ENDFOR
\STATE $j$ := random([1..k])
\STATE $\mathit{locals}$ :=  {\em $\overrightarrow{\mathit{locals}}[j]$}
\STATE{ // rest of the code $\ldots$}
\eindent
\end{algorithmic}
\end{algorithm}
\end{minipage}
\end{wrapfigure}

In general, a computation step of an object implementation is either
\begin{itemize}
\item
an invocation to a method of a base object, e.g.,
a register, which is assumed to be \emph{atomic}, or
\item
a \emph{send/receive} step in the context of a message-passing system, or
\item a local computation step on some set of local variables
(which cannot be accessed by other processes).
\end{itemize}
A computation step is called \emph{effect-free} if it is a local
computation step, or, if in the first case,
the invoked method itself is effect-free, e.g., a \textbf{Read} method of an atomic register, or if in the second case, it is a receive or a send of a message that does not modify the local state of the receiving process, e.g., sending a ``query'' message in the ABD register. For a preamble mapping $\cp$, we say that a method $M$ has an \emph{effect-free preamble} if all the computation steps up to and including $\cp(M)$ are effect-free. An object is said to have \emph{effect-free preambles} iff all its methods have effect-free preambles.

It can be easily proved that $O^k$ is equivalent to $O$, provided that
$O$ has effect-free preambles.
We also assume that the original tail strongly
linearizable objects are deterministic, i.e., they do not rely on
randomization.  Indeed, by definition, repeating the effect-free
preamble has no effect on local states of other processes.  Each
execution of $O^k$ can be transformed to an execution of $O$ where all
the preamble repetitions that are not ``used'' in an invocation (i.e.,
the value they compute is not selected to continue the computation)
can be simply removed. Since the original $O^k$ execution has exactly
the same history as the one of $O$, its linearizability w.r.t. the
specification of $O$ follows from the linearizability of the execution
of $O$. Conversely, every execution of $O$ can be transformed to an
execution of $O^k$ by ``appending'' sufficiently many repetitions of
the preamble and restricting the random choice to select the first
repetition.

\begin{theorem}\label{th:abd-star-lin}
For every object $O$ with effect-free preambles and $k\geq 1$, $O^k$ is equivalent to $O$.
\end{theorem}

\subsection{Quantifying the Blunting Power}\label{ssec:quantify}

We characterize the power of $O^k$ objects in lowering the probability that a program $P$ using them reaches some set $\outs$ of outcomes, compared to $P$ using the original objects $O$ instead. Since we interpret $\outs$ as ``bad'' states, lowering this probability is desirable.

For a set of objects $\objs$, $\objs^k$ is the set of objects $O^k$
with $O\in\objs$. While stating the result below, the program $P$ and
the set of outcomes $\outs$ are fixed (but arbitrary), and to simplify
the notation, we write $\probOf{\objs}$ instead of
$\probOf{P(\objs)\rightarrow\outs}$, for any set of objects
$\objs$. Also, we say that a program $P(\objs)$ has \emph{at most $r$
  $\mathsf{random}$ steps} if every execution of $P$ contains at most
$r$ steps corresponding to executing a $\mathsf{random}$
instruction.
This definition applies to programs using objects
$\objs$ and not the transformed objects $\objs^k$ which introduce
additional $\mathsf{random}$ steps.

We show that $\probOf{\objs^k}$ decreases with respect to
$\probOf{\objs}$
as the number of preamble iterations $k$ increases and exceeds
the maximum number $r$ of $\mathsf{random}$ steps in the program.
This provides a trade-off between time complexity, which grows with
$k$, and the probability of reaching bad outcomes, which decreases
with $k$.  This result is based on a worst-case analysis which makes
no assumptions about the structure of the program.

\begin{theorem}\label{th:main}
For every program $P(\objs)$ with $n\geq 1$ processes and at most $r\geq 1$ $\mathsf{random}$ steps, where $\objs$ is a set of tail strongly linearizable objects with effect-free preambles, set of outcomes $\outs$,
\begin{align*}
\probOf{\objs^k} \leq \probOf{\objs_a} + \left[1 - \left(\frac{\max\{0,k-r\}}{k}\right)^{n-1}\right] \cdot \left(\probOf{\objs}-\probOf{\objs_a}\right).
\end{align*}
\end{theorem}

Theorem~\ref{th:main} states that the probability of a bad outcome when
using objects in which the preamble is iterated $k$ times is at most the
probability when using atomic objects plus a fraction of the difference
between the probabilities when using atomic objects and when using the original
linearizable objects.
The fraction is, roughly speaking, the probability that the adversary
is able to manipulate the behavior to its advantage, and it
goes to 0 as $k$ increases, and thus the probability with the
preamble-iterated objects approaches the probability with atomic objects.

\subsection{Proof Outline for Theorem~\ref{th:main}}\label{ssec:proof}

We start by introducing some terminology.
The program $P(\objs^k)$ has two types of $\mathsf{random}$ instructions: the $\mathsf{random}$ instructions coming from the original program $P(\objs)$, which are outside of object implementations, and the $\mathsf{random}$ instructions added in the $\objs^k$ implementations (see Algorithm~\ref{algorithm:O-star}). The former are called \emph{program} $\mathsf{random}$ instructions, and the latter \emph{object} $\mathsf{random}$ instructions. Steps in an execution corresponding to program (object) $\mathsf{random}$ instructions are called program (object) $\mathsf{random}$ steps. Each method invocation in an execution of $P(\objs^k)$ performs $k$ iterations of a preamble (of some method of an object in $\objs$).
A preamble iteration is called \emph{randomization-free}
when it does \emph{not} overlap with a program $\mathsf{random}$ step, i.e., every program $\mathsf{random}$ step occurs either before or after all the steps of that preamble iteration.

Let $A$ be an adversary against $P(\objs^k)$ defining a probability distribution over executions/outcomes. Let $X$ be the event that \emph{all} the object $\mathsf{random}$ steps return indices that correspond to randomization-free preamble iterations. We decompose the probability of $A$ reaching a set of outcomes $\outs$ by conditioning on $X$:
\begin{align}\label{eq:cond1}
\probOf{P(\objs^k) || A\rightarrow\outs} = \probOf{( P(\objs^k) || A\rightarrow\outs)\ |\ X} \cdot \probOf{X} + \probOf{(P(\objs^k) || A\rightarrow\outs)\ |\ \neg X} (1-\probOf{X})
\end{align}

Lemma~\ref{lem:main1} (proved below)
shows that the probability of $A$ reaching $\outs$ conditioned on $X$
is upper bounded by the probability of any adversary reaching $\outs$
in the same program but with atomic objects instead of $\objs^k$.
That is,
$\probOf{( P(\objs^k) || A\rightarrow\outs)\ |\ X}\leq \probOf{P(\objs_a)\rightarrow\outs}$.
Lemma~\ref{lem:main2} (proved below)
shows that the probability of reaching $\outs$ with $\objs^k$
conditioned on $\neg X$ cannot be larger than
the probability of reaching $\outs$ with $\objs$,
i.e., $\probOf{ (P(\objs^k) || A\rightarrow\outs)\ |\ \neg X}
\leq \probOf{P(\objs)\rightarrow\outs}$.
Substituting
into (\ref{eq:cond1}), we get that
\begin{align}
\probOf{P(\objs^k) || A\rightarrow\outs} & \leq \probOf{P(\objs_a)\rightarrow\outs} \cdot \probOf{X} + \probOf{P(\objs)\rightarrow\outs} (1-\probOf{X})\label{eq:cond2} \\
& = \probOf{P(\objs_a)\rightarrow\outs} + \left(1 - \probOf{X}\right) (\probOf{P(\objs)\rightarrow\outs} - \probOf{P(\objs_a)\rightarrow\outs}) \nonumber
\end{align}
Lemma~\ref{lem:lb on prob of X} (proved below)
shows that $\probOf{X} \geq \left(\frac{\max\{0,k-r\}}{k}\right)^{n-1}$,
which
concludes the proof of the theorem.

\subsection{Detailed Proofs}\label{proofs}

\begin{lemma}\label{lem:main1}
$\probOf{( P(\objs^k) || A\rightarrow\outs)\ |\ X}\leq \probOf{P(\objs_a)\rightarrow\outs}$.
\end{lemma}

\begin{proof}
Based on the adversary $A$, we will define an adversary $A_\objs$
against $P(\objs)$ that mimics the adversary $A$ against $P(\objs^k)$
conditioned on $X$ for program $\mathsf{random}$ steps and takes the
``best'' choice for object $\mathsf{random}$ steps, i.e., the choice
that maximizes the probability of reaching $\outs$.
$A_\objs$ will cause all the prefixes of executions in
$E(A_\objs)$ that end with a program $\mathsf{random}$ step to be
complete w.r.t. each preamble mapping
of an object in $\objs$.
The construction of $A_\objs$ will ensure that
\begin{align}\label{eq:adv1}
\probOf{( P(\objs^k) || A\rightarrow\outs)\ |\ X}\leq \probOf{P(\objs) || A_\objs\rightarrow\outs}
\end{align}
Then, we will use the completeness w.r.t. preamble mappings of
execution prefixes to show that
\begin{align}\label{eq:adv2}
\probOf{P(\objs) || A_\objs\rightarrow\outs} \leq \probOf{P(\objs_a)\rightarrow\outs}.
\end{align}
which will complete the proof.  Details follow.

Given a sequence $\vec{v}$ of values returned by program $\mathsf{random}$ steps, let $\vec{u}$ be a sequence of values returned by program or object $\mathsf{random}$ steps such that $\vec{v}$ is a subsequence of $\vec{u}$ and for all index $i$ in $\vec{u}$ representing the value of an object $\mathsf{random}$ step,
\begin{align}\label{eq:adv3}
\probOf{( P(\objs^k) || A\rightarrow\outs)\ |\ X\ |\ \vec{u}[\leq i]} = max_{v\in\values} \probOf{( P(\objs^k) || A\rightarrow\outs)\ |\ X\ |\ \vec{u}[\leq i-1]\cdot v}
\end{align}
where $\probOf{( P(\objs^k) || A\rightarrow\outs)\ |\ X\ |\ \sigma}$ is the probability that $A$ reaches $\outs$ in $P(\objs^k)$ conditioned on $X$, and further conditioned on the fact that the first $|\sigma|$ $\mathsf{random}$ steps return the values in $\sigma$ (in the order defined by $\sigma$), and $\vec{u}[\leq i]$ is the prefix of $\vec{u}$ of length $i$ (by convention, $\vec{u}[\leq -1]$ is the empty sequence $\epsilon$). The schedule $A(\vec{u})$ contains $k$ preamble iterations for each method invocation, but only one of them, determined by the result of the object $\mathsf{random}$ step in that invocation, is used to continue the computation. Let $\mathsf{remRedundant}(A(\vec{u}))$ be the schedule where all the $k-1$ preamble iterations that are not used in a method invocation are removed. By the definition of the $\objs^k$ objects, $\mathsf{remRedundant}(A(\vec{u}))$ is a schedule producing a valid execution of $P(\objs)$. We define
$$
A_\objs(\vec{v})=\mathsf{remRedundant}(A(\vec{u})).
$$
By the construction, property (\ref{eq:adv3}) in particular, we have that property (\ref{eq:adv1}) holds. Also, since we consider schedules of $A$ conditioned on $X$, all the preamble iterations selected by object $\mathsf{random}$ steps are randomization-free, and therefore, at every program $\mathsf{random}$ step in $\mathsf{remRedundant}(A(\vec{u}))$, there is no invocation that started but did not finished its preamble.

To prove property (\ref{eq:adv2}), we show that there exists an adversary $A_{\objs_a}$ against $P(\objs_a)$ such that $\prob{P(\objs)}{A_\objs}=\prob{P(\objs_a)}{A_{\objs_a}}$. We rely on the facts that each object in $\objs$ is tail strongly linearizable, that tail strong linearizability is local (cf. Theorem~\ref{th:local}), and that all the prefixes of executions in $E(A_\objs)$ ending
with a program $\mathsf{random}$ step are complete w.r.t. each preamble mapping of an object in $\objs$.
The adversary $A_{\objs_a}$ is defined iteratively by enumerating
program $\mathsf{random}$ steps. Initially, by the definition of an
adversary, all the executions produced by $A_{\objs}$ are identical
until the first occurrence $rs_1$ of a program $\mathsf{random}$
step. By tail strong linearizability, it is possible to define a
valid linearization (satisfying each object specification)
of the invocations that started before $rs_1$ which does not depend on
execution steps that follow $rs_1$ (i.e., this linearization can be
extended by appending more invocations when considering steps after
$rs_1$). Let $\sigma_0$ be such a linearization. We will impose the
constraint that all the executions produced by $A_{\objs_a}$ start
with $\sigma_0$.

Next, we focus on execution prefixes that end just before the second occurrence $rs_2$ of a program $\mathsf{random}$ step. Assume that $rs_1$ is a random choice between a set of values $V$ and let $v\in V$. Using again the definition of an adversary, all the executions produced by the restriction of $A_{\objs}$ to the domain $v\cdot \mathbb{V}^*$
(sequences of values starting with $v$) are identical until $rs_2$.
By tail strong linearizability, there exists a linearization $\sigma_v$ of 
the invocations that started before $rs_2$ in these executions 
such that $\sigma_0$ is a prefix of $\sigma_v$. 
Moreover, $\sigma_v$ can be chosen in such a way that it does not depend on execution steps that follow $rs_2$.  We define $A_{\objs_a}$ such that $A_{\objs_a}(v\cdot \mathbb{V}^*) \in \sigma_v\cdot Act^*$ for each $v\in V$ ($Act$ denotes the set of call/return actions in a history). That is, each execution that the adversary produces when the first program $\mathsf{random}$ step returns $v$ starts with the linearization $\sigma_v$.

Iterating the same construction for all the remaining 
program $\mathsf{random}$ steps, 
we get an adversary $A_{\objs_a}$ against $P(\objs_a)$ such that 
$A_{\objs_a}(\vec{v})$ is a linearization of the invocations 
in $A_{\objs}(\vec{v})$,  for all $\vec{v}$. 
Therefore, $\prob{P(\objs)}{A_\objs}=\prob{P(\objs_a)}{A_{\objs_a}}$, 
and property (\ref{eq:adv2}) holds.
\end{proof}

\begin{lemma}\label{lem:main2}
$\probOf{ (P(\objs^k) || A\rightarrow\outs)\ |\ \neg X}\leq \probOf{P(\objs)\rightarrow\outs}$.
\end{lemma}
\begin{proof}
As in the proof of Lemma~\ref{lem:main1}, property (\ref{eq:adv1})
, one can define an adversary $A_\objs'$ against $P(\objs)$ that mimics the adversary $A$ against $P(\objs^k)$ for program $\mathsf{random}$ steps and takes the ``best'' choice for object $\mathsf{random}$ steps, i.e., the choice that maximizes the probability of reaching $\outs$. This argument is actually agnostic to the conditioning on $\neg X$, because it does not depend on
the specific results returned by object $\mathsf{random}$ steps from which to make a ``best'' choice.
 We include the conditioning only to match the proof goal coming from (\ref{eq:cond1}).
We have that
\begin{align}\label{eq:adv10}
\probOf{ (P(\objs^k) || A\rightarrow\outs)\ |\ \neg X}\leq \probOf{P(\objs) || A_\objs'\rightarrow\outs}
\end{align}

The result follows from the fact that $\probOf{ P(\objs) || A_\objs'\rightarrow\outs}\leq \probOf{P(\objs)\rightarrow\outs}$.
\end{proof}

\begin{lemma}\label{lem:lb on prob of X}
$\probOf{X} \geq \left(\frac{\max\{0,k-r\}}{k}\right)^{n-1}$.
\end{lemma}
\begin{proof}
Since the random choices in $\objs^k$ method invocations are independent, we have that
$
\probOf{X} = \prod_{i}\probOf{X_i}
$
where $X_i$ is the event that the $i$-th object $\mathsf{random}$ step in an invocation to a method of $\objs^k$ chooses a randomization-free preamble iteration (we assume an arbitrary but fixed total order on invocations in $P$).
The minimal value for $\probOf{X}$ can be attained by making many $\probOf{X_i}$ as small as possible.
To minimize the sum of $\probOf{X_i}$ terms, we need that each $\mathsf{random}$ step overlaps with a maximum number of preamble iterations, i.e., one preamble iteration from each other process. Then, to maximize the number of small $\probOf{X_i}$ terms, we need to maximize the number of invocations that contain a maximal number of preamble iterations overlapping with a $\mathsf{random}$ step.
These two constraints can be attained assuming that all 
program $\mathsf{random}$ steps are in the same process 
and each one of them overlaps with a different preamble 
iteration from the same invocation of each other process.
If $k \le r$, the adversary can ensure that no object random step
returns an index that corresponds to a randomization-free preamble
iteration, which is the reason for the use of the max function.
Therefore,
for $n-1$ invocations $i$,
\begin{align*}
\probOf{X_i}\ =\ \frac{\max\{0,k-r\}}{k}
\end{align*}
and $\probOf{X_j}=1$ for the rest of the invocations $j$.
Therefore,
\begin{align*}
\probOf{X} \geq \left(\frac{\max\{0,k-r\}}{k}\right)^{n-1}
\end{align*}
\end{proof}

\section{Examples of Tail Strongly Linearizable Objects}\label{app:tasl_objects}

We discuss several objects introduced in the literature that
are \emph{not} strongly linearizable,
but are tail strongly linearizable with respect to
some non-trivial, effect-free preamble mapping.

\subsection{ABD Register}

Variations of the ABD implementation of a register in a
crash-prone message-passing system are used in many applications.
Unfortunately, it is impossible to have a strongly
linearizable version of ABD~\cite{AttiyaEW2021,ChanHHT2021arxiv}.
However, as we show next, our transformation is applicable to ABD.

Specifically, we show that the multi-writer variant~\cite{LynchS1997} of the
ABD register~\cite{AttiyaBD1995}
(which is listed in Algorithm~\ref{algorithm:ABD} in the appendix
and explained in the introduction)
is tail strongly linearizable w.r.t. the preamble mapping $\cp_{ABD}$
that associates \textbf{Read} and \textbf{Write} with the control
points Lines~\ref{line:ReadControl} and~\ref{line:WriteControl},
respectively.  These are the control points of the steps that assign
the return value of queryPhase to $(v,u)$ and $(-,(t,-))$,
respectively.

\begin{theorem}\label{ABD:TSL}
The ABD object in Algorithm~\ref{algorithm:ABD} is tail strongly linearizable w.r.t.~$\cp_{ABD}$.
\end{theorem}
\begin{proof}
The timestamp of a \textbf{Read} invocation is the timestamp returned by its query phase (the value $u$ at line~\ref{line:ReadControl}), and the timestamp of a \textbf{Write} is the timestamp given as parameter to its update phase (the pair $(t+1,i)$ at line~\ref{line:WriteUpdate}). The timestamp of an invocation $o$ is denoted by $\mathsf{ts}(o)$.

Given an execution $e$ that is complete w.r.t. $\cp_{ABD}$, we say that an invocation $o$ is \emph{logically-completed} in $e$ when there exists an invocation $o'$ that returns in $e$ such that $\mathsf{ts}(o)\leq \mathsf{ts}(o')$. Since $o$ and $o'$ may coincide, if an invocation returns in $e$, then it is also logically-completed in $e$. By definition, every invocation in $e$ has a well-defined timestamp (since every invocation passed the query phase).

We define a function $f$ that associates to each such execution $e$ a linearization that contains all the invocations that are logically-completed in $e$ ordered according to their timestamp. A set of invocations in $e$ that have the same timestamp consists of exactly one \textbf{Write} invocation and some number of \textbf{Read} invocations. The linearization $f(e)$ orders the write before all the reads with the same timestamp, if any.

To show that $f$ is prefix-preserving, let $e,e'\in E(\text{ABD},\cp)$ such that $e$ is a prefix of $e'$. We show that a linearization of $e$ where invocations that are logically-completed in $e$ are ordered before invocations that are \emph{not} logically-completed is
consistent with an analogous linearization of $e'$.

For an invocation $o_1$ that is logically-completed in $e$, we show that $\mathsf{ts}(o_1)\leq \mathsf{ts}(o_2)$ for every invocation $o_2$ that is not logically-completed in $e$.
There are two cases to consider.
First, if $o_2$ queries after $e$, then we use the fact that ABD guarantees that the timestamp of an invocation is smaller than or equal to the timestamp returned by any query phase starting after that invocation returned. By the definition of logically-completed, there exists an invocation $o_1'$ that returns in $e$ such that $\mathsf{ts}(o_1)\leq \mathsf{ts}(o_1')$. Using the property of ABD mentioned above, we get that $\mathsf{ts}(o_1')\leq \mathsf{ts}(o_2)$, which implies that $\mathsf{ts}(o_1) \leq \mathsf{ts}(o_2)$.
Second, if $o_2$ queries during $e$, then by the definition of logically-completed, $\mathsf{ts}(o_2) > \mathsf{ts}(o_2')$ for every invocation $o_2'$ that returns in $e$. Since $o_1$ is logically-completed in $e$, we get that  there exists an invocation $o_1'$ that returns in $e$ such that $\mathsf{ts}(o_1)\leq \mathsf{ts}(o_1')$. Therefore, $\mathsf{ts}(o_1) \leq \mathsf{ts}(o_2)$.
Next, we show that there cannot exist a \textbf{Write} invocation $o_1$ that is \emph{not} logically-completed in $e$ while a \textbf{Read} invocation $o_2$ with the same timestamp is logically-completed in $e$. Clearly, $o_1$ cannot query after $e$ since $o_2$ queries during $e$ by definition. Assuming that both invocations query during $e$, we get a contradiction because the definition of logically-completed implies that $\mathsf{ts}(o_1) > \mathsf{ts}(o_1')$ for every invocation $o_1'$ that returns in $e$ and there exists an invocation $o_2'$ that returns in $e$ such that $\mathsf{ts}(o_2)\leq \mathsf{ts}(o_2')$. These two statements imply that $\mathsf{ts}(o_1) > \mathsf{ts}(o_2)$ which is a contradiction to the fact that $o_1$ and $o_2$ have the same timestamp.

Finally, note that an invocation $o_1$ that is \emph{not} logically-completed in $e$ cannot return before an invocation $o_2$ that is logically-completed in $e$. Since $o_2$ queries during $e$, this would imply that $o_1$ returns in $e$ which would imply that $o_1$ is logically-completed in $e$.
\end{proof}

The above result holds also for the original single-writer
version~\cite{AttiyaBD1995}, which is also not strongly
linearizable~\cite{HadzilacosHT2020arxiv4,ChanHHT2021arxiv}.

\subsection{Snapshot}

Another popular shared object is the atomic snapshot.
It is impossible to implement a strongly-linearizable lock-free snapshot
object using
single-writer registers~\cite{HelmiHW2012} and it is impossible to
implement a strongly-linearizable wait-free snapshot object using
multi-writer registers~\cite{DenysyukW2015}.
However, we show next that we can apply our transformation to
the linearizable wait-free snapshot implementation in~\cite{AfekADGMS1993},
which uses single-writer registers.

The snapshot object implementation of~\cite{AfekADGMS1993} uses an array \texttt{M} of registers whose length is the number of processes (accesses to these registers are atomic, i.e., they execute instantaneously). It provides a $\mathsf{Scan}()$ method that returns a snapshot of the array and an $\mathsf{Update}(v)$ method by which
a process $i$ writes value $v$ in \texttt{M}[$i$]. $\mathsf{Scan}$ performs a series of \emph{collects}, i.e., successive reads of the array's cells in some fixed order; a collect in a process can interleave with steps of other processes. This series of collects stops when either two successive collects return identical values,
or the process observes that another process has executed at least two $\mathsf{Update}$ invocations during the timespan of the $\mathsf{Scan}$. In the latter case, the return value is the last snapshot written by the other process during an $\mathsf{Update}$. An $\mathsf{Update}$ invocation at a process $i$ starts with a $\mathsf{Scan}$ followed by an atomic write to \texttt{M}[$i$] of the result of $\mathsf{Scan}$ together with the value received as argument (and a local sequence number seq$_i$ that is read in other $\mathsf{Scan}$ invocations).

This snapshot object is known to \emph{not} be strongly linearizable~\cite{GolabHW2011}, but it is tail strongly linearizable w.r.t. a preamble mapping that maps each $\mathsf{Scan}$ to the control point just before it returns and each $\mathsf{Update}$ to the initial control point. The linearization associated to an execution that is complete w.r.t. this preamble mapping contains all the (possibly pending) $\mathsf{Scan}$ invocations and all the $\mathsf{Update}$ invocations that performed their writes to the array cells, in some order consistent with the specification (each $\mathsf{Scan}$ is linearized after an $\mathsf{Update}$ if it observes its value). Actually, the preamble of $\mathsf{Update}$ can be defined in an arbitrary manner, e.g., extended until the end of its scan, and tail strong linearizability would still hold.
The reason is that
an $\mathsf{Update}$ is linearized only if it executed its write---
the scan it performs before the write is only to ensure progress (wait-freedom).
As can be seen in Section~\ref{section:ABD}, 
extending a preamble may help in reducing the probability 
of reaching ``bad'' outcomes, 
but this comes at a cost in terms of time complexity.

\subsection{Multi-Writer Multi-Reader Register}

Another central shared object is a multi-writer multi-reader register.
There is no strongly-linearizable wait-free implementation
of such a register using single-writer registers~\cite{HelmiHW2012}.
We show, however, that our transformation can be applied to the
linearizable implementation in~\cite{VitanyiA1986}.

In this implementation, each value written has a timestamp, which
is a pair consisting of an integer and a process identifier.
A single-writer register \textit{Val}[$i$] is
associated with each writer $i$ of the implemented register.
When a read is invoked on the implemented register, the reader reads
(value, timestamp) pairs from all the \textit{Val} registers, chooses the
value with the largest timestamp using lexicographic ordering, and returns
that value.
When a write of value $v$ on the implemented register is invoked at
writer $i$, the writer calculates a new timestamp and writes the value
together with the timestamp into \textit{Val}[$i$].
To calculate the new timestamp, $i$ reads all the \textit{Val}
variables and extracts from it the timestamp entry.
Its new timestamp is one plus the maximal timestamp of of all
other processes, together with its identifier.
This implementation is tail strongly linearizable by choosing
the preamble of the read method to end just before it returns
and the preamble of the
write method to end immediately before writing to \textit{Val}[$i$].
The tail strong linearizability proof is similar to the one for the ABD
register.

\subsection{Single-Writer Multi-Reader Register}

Yet another standard shared object is a (single-writer) multi-reader
register.
A well-known implementation of such a register using (single-writer)
single-reader registers is given in~\cite{IsraeliL1993}.
This implementation is not strongly linearizable, which can be shown
by mimicking the counter-example for the ABD register appearing
in~\cite{HadzilacosHT2020arxiv4}.
However, our transformation is applicable to this implementation, as
we show next.
(It seems likely that the argument in~\cite{ChanHHT2021arxiv} can be
adapted to show that it impossible to have a strongly-linearizable
implementation of a multi-reader register using single-reader registers,
as it is easy to simulate a message-passing channel with a single-reader
register.)

In the implementation, a single-reader register {\em Val[i]} is associated
with each reader $i$ of the implemented register.
To write a value $v$ to the implemented register, the (unique) writer
writes $v$, together with a sequence number, into all of the
{\em Val} registers.
The readers communicate with each other via a (two-dimensional) array {\em Report} of
single-reader registers, where reader $i$ writes to all the registers in
row $i$ and reads from all the registers in column $i$.
When a read of the implemented register is invoked at process $i$,
it reads (value, sequence number) pairs from {\em Val[i]} and from
all the registers in column $i$ of {\em Report}; it then
chooses the value to return with the largest sequence number, writes
this pair to all the registers in row $i$ of {\em Report}, and returns.
This implementation is tail strongly linearizable:  the
preamble of the read method ends just before the first write to
an element of {\em Report}, while the preamble of the write
method is empty.
As before, the proof of tail strong linearizability is similar to the
one for the ABD register.

\section{Related Work}
\label{section:related}

Golab, Higham and Woelfel~\cite{GolabHW2011} were the first
to recognize the problem when linearizable
objects are used with randomized programs, via an
example using the snapshot object implementation of~\cite{AfekADGMS1993}.
They proposed \emph{strong linearizability}
as a way to overcome the increased vulnerability
of programs using linearizable implementations to strong adversaries,
by requiring that the linearization order of operations
at any point in time be consistent with the linearization order
of each prefix of the execution.
Thus, strongly-linearizable implementations limit the adversary's
ability to gain additional power by manipulating
the order of internal steps of different processes.
Consequently, properties holding when a
concurrent program is executed with an atomic object,
continue to hold when the program is executed with a strongly-linearizable
implementation of the object.
Strong linearizability is a special case of our class of implementations,
where the preamble of each operation is empty and thus, vacuously, effect-free;
in this case, applying the preamble-iterating transformation results in no
change to the implementation.

Other than~\cite{AttiyaEW2021,ChanHHT2021arxiv}
which studied message-passing implementations, prior work on strong
linearizability focused on implementations using shared objects,
and considered various progress properties.
If one only needs {\em obstruction-freedom}, which requires an operation to
complete only if it executes alone, any object can be implemented
using single-writer registers~\cite{HelmiHW2012}.
When considering the stronger property of {\em lock-freedom} (or
{\em nonblocking}),
which requires that as long as some operation is pending, some
operation completes, single-writer registers are not sufficient
for implementing multi-writer registers, max registers, snapshots,
or counters~\cite{HelmiHW2012}.
If the implementations can use multi-writer registers, though, it
is possible to get lock-free implementations of max registers,
snapshots, and monotonic counters~\cite{DenysyukW2015}, as well as of
objects whose operations commute or overwrite~\cite{OvensW2019}.
It was also shown~\cite{AttiyaCH2018} that there is no
lock-free implementation of a queue or a stack from
objects whose readable versions have consensus number less than the
number of processes, e.g., readable test\&set.
For the even stronger property of {\em wait-freedom}, which requires every
operation to complete, is is possible to implement bounded max registers
using multi-writer registers~\cite{HelmiHW2012}, but it is impossible
to implement max registers, snapshots, or monotonic
counters~\cite{DenysyukW2015} even with multi-writer registers.
The bottom line is that the only known strongly-linearizable wait-free
implementation is of a bounded max register (using multi-writer registers),
while many impossibility results are known.

\emph{Write strong linearizability (WSL)}~\cite{HadzilacosHT2021} is a
weakening of strong linearizability designed specifically for register
objects. It requires that executions be mapped to linearizations where
only the projections onto write operations are prefix-preserving.  While
single-writer registers are trivially WSL, neither the original multi-writer
ABD nor the preamble-iterating version we introduce in this paper is
WSL~\cite{HadzilacosHT2021}.
The WSL implementation given in~\cite{HadzilacosHT2021} has effect-free
preambles, and so our transformation is applicable to it.
It is not known whether it
is possible to implement WSL multi-writer registers in crash-prone
message-passing systems.

Our approach draws (loose) inspiration from the vast research on
\emph{oblivious RAM} (\emph{ORAM}) (initiated in~\cite{GoldreichO1996}),
although the goals and technical details significantly differ.
ORAMs provide an interface through which a program can hide
its memory access pattern, while at the same time accessing
the relevant information.
More generally, \emph{program obfuscation}~\cite{BarakGIRSVY2012}
tries to hide (obfuscate) from an observer knowledge about the
program's functionality, beyond what can be obtained from its
input-output behavior.
The goal of ORAMs and program obfuscation is to hide information
from an adversary,
while our goal is to blunt the adversary's ability to
disrupt the program's behavior by exploiting linearizable implementations
used by the program.
We borrow, however, the key idea of introducing additional
randomization into the implementation,
in order to make it less vulnerable to the adversary.

\section{Discussion}

We have presented the preamble-iterating transformation for
a variety of linearizable object implementations,
e.g.,~\cite{AttiyaBD1995,AfekADGMS1993,IsraeliL1993,VitanyiA1986},
which approximately preserves the probability
of reaching particular outcomes,
when these implementations replace the corresponding atomic objects.
In this manner, it salvages randomized programs that use
these highly-useful objects---which do not have strongly-linearizable
implementations---so they still terminate,
without modifying the programs or their correctness proofs.
Furthermore, the transformation is mechanical, once the preamble is identified.

Our results are just the first among many new opportunities
for modular use of object libraries in randomized concurrent programs,
including the following exciting avenues for future research.

One direction is to improve our analysis and obtain better bounds,
specifically, by exploring the tradeoff between the increased
complexity of many repetitions of the preamble,
and decreased probability of bad outcomes.

It is also crucial to reduce the number of random steps considered
in the analysis, and at least, to bound them.
This can be done by making assumptions about the structure of the
randomized concurrent program.
For example, many randomized programs are round-based,
where each process takes a fixed (often, constant) number $s$
of random steps in each round,
and termination occurs with high probability
within some number of rounds, say $T$.
In this case, we can let the program run for $T$ rounds
and apply the preamble-iterating transformation with $k > T \cdot s$;
if the program does not terminate within $T$ rounds,
which happens with small probability,
the program just continues with the original, linearizable object.
An alternative approach for dealing with an unbounded number
of random steps is
to assume that the rounds are \emph{communication-closed}~\cite{DBLP:journals/scp/ElradF82},
resulting in a smaller number of random choices that could affect
the linearizable implementation.

Another direction is to consider other objects without wait-free
strongly-linearizable implementations,
e.g., queues or stacks~\cite{AttiyaCH2018},
which lack effect-free preambles that can be easily repeated.
For such objects, it might be possible to \emph{roll back} the effects
of repeating certain parts of their implementation.

\bibliographystyle{plain}
\bibliography{citations}

\appendix
\begin{algorithm}[tb]
\caption{ABD simulation of a multi-writer register in a message-passing system.}
\label{algorithm:ABD}
\begin{algorithmic}[1]
\STATE {\bf local variables:}
\STATE  {\em sn}, initially 0
    \COMMENT{for readers and writers, used to identify messages}
\STATE {\em val}, initially $v_0$
    \COMMENT{for servers, latest register value}
\STATE  {\em ts}, initially $(0,0)$
    \COMMENT{for servers, timestamp of this value, (integer, process id) pair}
\vspace*{.1in}
\begin{multicols*}{2}
\STATE {\bf function queryPhase():}
\STATE  {\em sn}++
\STATE  broadcast $\langle$"query",{\em sn}$\rangle$
\STATE  wait for $\geq \frac{n+1}{2}$ reply msgs to this query msg
\STATE  ({\em v,u}) $:=$ pair in reply msg with largest timestamp
\STATE  return ({\em v,u})
\vspace*{.1in}
\STATE {\bf when $\langle$"query",{\em s}$\rangle$ is received from $q$:}
\STATE  send $\langle$"reply",{\em val,ts,s}$\rangle$ to $q$
\vspace*{.1in}
\STATE {\bf function updatePhase({\em v,u}):}
\STATE  {\em sn}++
\STATE  broadcast $\langle$"update",{\em v, u, sn}$\rangle$
\STATE  wait for $\geq \frac{n+1}{2}$ ack msgs for this update msg
\STATE  return
\vspace*{.1in}
\STATE {\bf when $\langle$"update",{\em v,u,s}$\rangle$ is received from $q$:}
\STATE  if $u > ts$ then ({\em val},{\em ts}) $:=$ ({\em v,u})
\STATE  send $\langle$"ack",{\em s}$\rangle$ to $q$
\vspace*{.3in}
\STATE {\bf Read():}  \label{start-line}
\STATE  ({\em v,u}) $:=$ queryPhase()  \label{line:ReadControl}
\STATE  updatePhase({\em v,u})  \COMMENT{write-back}
\STATE  return {\em v}
\vspace*{.1in}
\STATE {\bf Write({\em v}) for process with id $i$:}
\STATE  $(-,(t,-)) :=$ queryPhase()  \COMMENT{just need integer in time\-stamp} \label{line:WriteControl}
\STATE  updatePhase({\em v},$(t+1,i)$) \label{line:WriteUpdate}
\STATE  return   \label{end-line}
\end{multicols*}
\end{algorithmic}
\end{algorithm}

\section{Case Study with ABD}
\label{section:overview}

This appendix presents a detailed case study of the benefits of
our preamble-iterating transformation when the program appearing in
Algorithm~\ref{algorithm:simplified-weakener} (presented in
Section~\ref{section:intro}) uses ABD registers.
This program is a simplified version of the {\em weakener}
program~\cite{HadzilacosHT2020randomized}, restricted to
only three processes, $p_0$, $p_1$, and $p_2$, that execute a single round.
We show in Section~\ref{ssec:weakener:atomic} that $p_2$ terminates with
probability at least 1/2 when the program uses atomic registers.
In contrast, we show in Section~\ref{ssec:app_ABD} that a strong
adversary can force $p_2$ to loop forever when the registers are
implemented using ABD.
Hadzilacos, Hu and Toueg~\cite{HadzilacosHT2020randomized}
showed that termination is
prevented in the weakener algorithm if the adversary has free rein to
choose the linearization points of the registers used; our example shows
an explicit execution using ABD that fails to terminate.
Since ABD is a tail strongly linearizable object with read-only preambles,
Theorem~\ref{th:main} implies that using ABD$^2$ in the program ensures
termination of $p_2$ with probability at least 1/8.
(ABD$^2$ is the special case of Algorithm~\ref{algorithm:ABD-star}
when $k = 2$; Algorithm~\ref{algorithm:ABD-star} is the result
of applying the transformation in Algorithm~\ref{algorithm:O-star}
to ABD, given in Algorithm~\ref{algorithm:ABD}.)
Section~\ref{subsec:abd2-prob} is devoted to a specialized analysis
that improves on the generic result and shows that $p_2$ terminates with
probability at least 3/8, indicating that there can be room for
improvement in our quantitative analysis.

\subsection{Success Probability with Atomic Registers}
\label{ssec:weakener:atomic}

We argue that with probability at least 1/2, process $p_2$ terminates,
when the program is using atomic registers;
this implies the same property when the program is composed with
strongly-linearizable registers (cf. Theorem~\ref{th:strong_lin}).

Let $u_1$, $u_2$, and $c$ be the values used in the test on
Line~\ref{testline}.
If $u_1 = \bot$, or $u_2 = \bot$, or $c = -1$, then
the test fails 
and $p_2$ terminates.

Suppose $u_1 \ne \bot$, $u_2 \ne \bot$, and $c \ne -1$.
Then at least one of $p_0$'s and $p_1$'s writes to $R$ precedes
$p_2$'s first read of $R$, and $p_0$'s write to $C$ precedes
$p_2$'s read of $C$.
If both writes to $R$ (by $p_0$ and $p_1$) precede
$p_2$'s first read of $R$, then $u_1 = u_2$, implying that the test
fails, since the common value cannot be equal to both $c$ and
$1-c$.  Therefore $p_2$ terminates.

Without loss of generality, suppose $p_0$'s write to $R$ precedes
$p_1$'s write to $R$.
Then the remaining situation is that $p_0$'s write to $R$ precedes
$p_2$'s first read of $R$ (so $u_1 = 0$), which precedes $p_1$'s
write to $R$, which precedes $p_2$'s second read of $R$ (so $u_2 = 1$).
With probability 1/2, $p_0$ writes 1 into $C$, which is then read by $p_2$,
so $c = 1$.  The test fails
since $u_1 = 0$, which is not equal to $c = 1$,
and $p_2$ terminates.

Thus, in the only situation in which $p_2$ is not guaranteed to terminate,
the probability of $p_2$ terminating is 1/2.

\subsection{Zero Success Probability with ABD Registers}
\label{ssec:app_ABD}

Next we explain how a strong adversary can force $p_2$ to loop forever
when the program uses linearizable registers,
and in particular, ABD registers, instead of atomic registers.
Each process runs a separate instance of ABD for each of the
shared variables $R$ and $C$.
(See Algorithm~\ref{algorithm:ABD}.)

\begin{figure}[t]
\includegraphics[scale=.35]{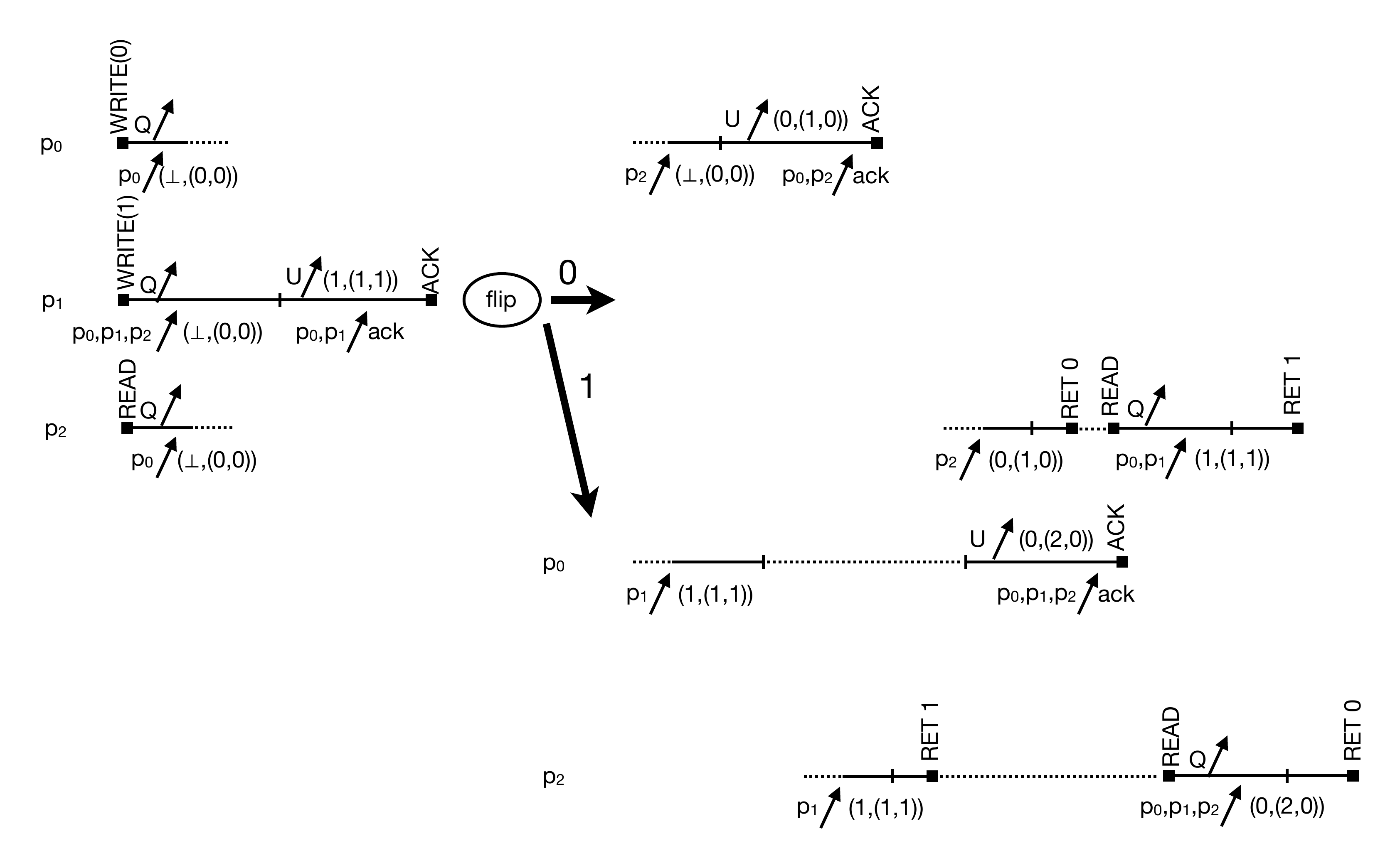}
\vspace{-2mm}
\caption{A strong adversary against ABD registers.}
  \label{fig:adv}
\vspace{-2mm}
\end{figure}

Figure~\ref{fig:adv} illustrates the counter-example,
just focusing on the Reads and Writes on $R$.
Time increases left to right.
The upper two timelines to the right of the coin flip indication
show the extensions of $p_0$'s and $p_2$'s computations when the flip
returns 0, while the lower two timelines show them when the flip
returns 1.  (There are no timelines for $p_1$ after the flip as it
does not access $R$ any more.)
Arrows leaving a timeline indicate broadcasts of query messages (labeled Q)
and update messages (labeled U and including the data), while
arrows entering a timeline indicate reply messages (containing data)
and ack messages received and are labeled with the senders.
Irrelevant update phases are not included.

Suppose that $p_0$ invokes its Write of 0 on $R$.
Let it receive the first reply to its query from $p_0$ (itself) containing
value $\bot$ and timestamp $(0,0)$.
Concurrently, suppose that $p_1$ invokes its Write of 1 on $R$.
Let it receive replies to its query from all the processes containing
value $\bot$ and timestamp $(0,0)$.
Then $p_1$ broadcasts its update message with value 1 and timestamp $(1,1)$.

Then suppose $p_2$ invokes its first Read of $R$.
Let it receive a reply to its query from $p_0$ with value $\bot$
and timestamp $(0,0)$.  That is, $p_0$ has not yet received $p_1$'s update
message when it replies to $p_2$.

Then suppose $p_1$ gets acks from $p_0$ and $p_1$ and completes its Write.
Then $p_1$ flips the coin.
In both cases discussed next, the adversary ensures that $p_2$ reads $C$ after
$p_1$ writes $C$ so that $p_2$'s local variable $c$ contains the
result of the coin.

{\em Case 1:}  Suppose the coin returns 0.  The adversary extends the
execution as follows, to ensure that $p_2$'s pending Read returns
0 into $p_2$'s local variable $u_1$ and $p_2$'s second Read returns 1
into $p_2$'s local variable $u_2$, causing $p_2$ to pass
the test at Line~\ref{testline} and loop forever.
Recall that $p_0$'s Write is also still pending.

Suppose $p_0$ gets its second reply from $p_2$ with value $\bot$ and
timestamp $(0,0)$; i.e., $p_2$ has not yet received $p_1$'s update message
when it replies to $p_0$.
Then $p_0$ broadcasts its update message with value 0 and timestamp $(1,0)$,
receives acks from $p_0$ and $p_2$, and completes its Write.

Now suppose that $p_2$ gets a reply from $p_2$ (itself) with value 0 and
timestamp $(1,0)$; i.e., $p_2$ has already received $p_0$'s update
message when it replies to itself.  So $p_2$ chooses value 0 and
timestamp $(1,0)$ for the update and its Read returns~0.

Then $p_2$ invokes its second Read of $R$.
Suppose that, in response to its query, it receives replies from $p_0$
and $p_1$ with value 1 and timestamp 1.
So $p_2$ chooses value 1 and timestamp $(1,1)$ for the update and
its Read returns 1.

{\em Case 2: } Suppose the coin returns 1.  The adversary extends the
execution as follows, to ensure that $p_2$'s pending Read returns
1 into $p_2$'s local variable $u_1$ and $p_2$'s second Read returns 0
into $p_2$'s local variable $u_2$, causing $p_2$ to pass 
the test at Line~\ref{testline} and loop forever.
Recall that $p_0$'s Write is also still pending.

Suppose $p_0$ gets its second reply from $p_1$ with value 1 and timestamp $
(1,1)$.
Then $p_2$ gets its second reply from $p_1$ with value 1 and timestamp $(1,1)$.
So $p_2$ chooses value 1 and timestamp $(1,1)$ for the update and its
Read returns 1.

We go back to considering $p_0$'s pending Write.  Next $p_0$
broadcasts its update message with value 0 and timestamp (2,0).
It receives ack messages from all three processes and the Write
finishes.

Finally, $p_2$ invokes its second Read of $R$.
Suppose that, in response to its query message,
$p_2$ receives reply messages with value 0 and timestamp
(2,0) from both $p_0$ and $p_1$.
So $p_2$ chooses value 0 and timestamp (2,0) for the update and its
Read returns 0.

\subsection{Blunting the Adversary with $ABD^2$}
\label{subsec:abd2-prob}

Now we consider the result of executing
Algorithm~\ref{algorithm:simplified-weakener}
using shared registers that are implemented with ABD$^2$.
We first give a simple argument, based on our main theorem,
that $p_2$ terminates with probability at least 1/8.
Then we show through a more specialized argument that this bound
is at least 3/8.

\begin{algorithm}[tb]
\caption{The transformed version ABD$^k$ corresponding to
ABD in Algorithm~\ref{algorithm:ABD}.}
\label{algorithm:ABD-star}
\begin{algorithmic}[0]
\vspace{-4mm}
\begin{multicols*}{2}
\STATE method {\bf Read():}
\bindent
\FOR{{\em i} := 1 \TO k }
\STATE{({\em v[i],ts[i]}) := queryPhase()}
\ENDFOR
\STATE {\em j} := random([1..k])
\STATE ({\em v,u}) := ({\em v[j],ts[j]})
\STATE  updatePhase({\em v,u})   // write-back
\STATE  \textbf{return} {\em v}
\eindent
%
%
\STATE method {\bf Write({\em v})} for process with id {\em p}:
\bindent
\FOR{{\em i} := 1 \TO k }
\STATE{ (-,({\em t[i]},-)) := queryPhase()   }
\ENDFOR
\STATE {\em j} := random([1..k])
\STATE {\em t} :=  {\em t[j]}
\STATE updatePhase({\em v},({\em t}+1,{\em p}))
\STATE  \textbf{return}
\eindent
\end{multicols*}
\vspace{-4mm}
\end{algorithmic}
\end{algorithm}

\subsubsection{A lower bound on the probability of $p_2$ terminating.}
As seen in Section~\ref{ssec:app_ABD}, the ABD register is ``exploited'' by the
adversary by scheduling the coin-flip \emph{during} the query phases performed in $p_0$'s Write
and $p_2$'s Read (in order to schedule some replies only when the result of the coin-flip is known).
Actually, scheduling the coin-flip so that it does \emph{not} overlap with a query phase
(it occurs after or before any query phase in a concurrently executing invocation), provides no gain to the adversary w.r.t. the atomic register case.
Indeed, in such a scenario, the linearization order between invocations that completed their query phase before the coin-flip is fixed even if they are still pending, a property that we call tail strong linearizability (see Theorem~\ref{ABD:TSL}), and the adversary cannot change the linearization order between the writes in particular, to accommodate a specific result of the coin flip.

When using $ABD^2$, the adversary can schedule the coin-flip to overlap with one of the two query phases in $p_0$'s Write
and $p_2$'s Read, but with probability 1/4 both of these invocations will choose to adopt the value-timestamp pair returned by the other query phase that does not overlap with the coin-flip (each invocation makes this choice with probability 1/2 and these choices are independent). Therefore, $ABD^2$ can blunt the adversary and with probability 1/4 make it behave as in the atomic register case. Therefore, with probability at least $1/4 \cdot 1/2=1/8$, the process $p_2$ terminates. This lower bound is a particular instance of our main result stated in Theorem~\ref{th:main}.

\subsubsection{A more detailed analysis.}
The reasoning above was agnostic to the particular values written to the registers or the conditions that are checked in the program. This makes it extensible to arbitrary programs and objects as we show in Section~\ref{section:ABD}. Nevertheless, as expected, a more precise analysis that takes into account these specifics can derive a better (bigger) lower bound. We present such an analysis in the following, showing that $p_2$ terminates with probability at least 3/8.

We show that no adversary can cause $p_2$ to loop forever, i.e.,
pass 
the test at Line~\ref{testline}, with probability more than 5/8.

Let $W_0$ be $p_0$'s Write of 0 to $R$,
$W_1$ be $p_1$'s Write of 1 to $R$,
$R_1$ be $p_2$'s first Read of $R$, and
$R_2$ be $p_2$'s second Read of $R$.

Consider the set $\mathcal E$ of executions that end with the program
coin flip by $p_1$ (Line~\ref{line:sim-weak-flip} in
Algorithm~\ref{algorithm:simplified-weakener}) and thus contain all of
$W_1$.  An adversary defines a probability distribution over $\mathcal
E$ which is a mapping $D$ from executions in $\mathcal E$ to
probabilities that sum to 1.
We refer to the adversary as {\em winning} when it causes $p_2$ to loop forever,
which happens only if $R_1$ reads the same value as the program
coin flip and $R_2$ reads the opposite value.

We show that the contribution of each execution
$E\in {\mathcal E}$ to the adversary's probability of winning is at
most $5 \cdot D(E) / 8$ (the sum over all $E$ leads to the $5/8$
bound). This proof considers a number of cases depending on which and
how many query phases of $W_0$ and $R_1$ finished in $E$.

When both query phases of either $W_0$ or $R_1$ are finished in $E$, the contribution of $E$ is actually at most
$D(E)/2$ (\textbf{Case 1} and \textbf{Case 2}). Since the random choice of which query response to use is independent
of the program coin flip, and the query responses are fixed before the coin flip,
the probability that the adversary wins in continuations of $E$ is at most 1/2.
In essence, Cases 1 and 2 behave as in the atomic case, since the
read-only preamble is already finished before the coin flip.
Otherwise, if the first query phase of $R_1$ did not yet return in $E$ (\textbf{Case 3}),
the value returned by one of  $W_0$'s query phases does not
``depend'' on the program coin flip ($W_0$'s first query returns in $E$ or otherwise, $W_0$'s second query returns $1$).
When choosing the result of this query in $W_0$,
the adversary fails for at least one value of the program coin flip
(wins with probability at most 1/2).
When choosing the other query, the adversary wins with probability at most 3/4, more precisely, at most 1/2 for one value of the program coin flip.
This is due to the random choice about which query response to return in $R_1$.
Overall, splitting over the random choice in $W_0$, we get that the adversary can win in continuations of
$E$ with probability at most (1/2 + 3/4)/2 = 5/8. Finally, for the case where the first query phase of $R_1$ returns in $E$ (\textbf{Case 4}),
the adversary's best strategy is to let this query phase return value 1 (written by $W_1$). However,
it can win with probability at most 1/2. When the program coin flip returns 0, if $R_1$ returns 0 because it chooses the response
of a second query phase, $R_2$ will return 0 as well
(since $W_0$ must have been linearized after $W_1$),
which means that the adversary fails in all such continuations.

\medskip
\noindent
\textbf{Case 1:}
Consider an execution $E \in {\mathcal E}$ such that
both query phases of $W_0$ are already finished.

If the random choice for which query phase result to use
for $W_0$ is included in $E$, then the timestamps of $W_0$ and $W_1$ are fixed,
as is their linearization order.
Without loss of generality, suppose $W_0$ is linearized before $W_1$;
then by linearizability, $p_2$ reads either 0,0 or 0,1 or 1,1, but not
1,0.  In order for the adversary to win, $p_2$ must read 0,1 {\em and}
the program coin flip must return 0, which occurs with probability 1/2.
Thus the contribution of $E$ to the adversary's probability of winning
is at most $D(E)/2$.

Suppose the random choice for which query phase result to use for $W_0$
is not included in $E$.
The continuations of $E$ can linearize $W_0$ before $W_1$ with some
probability $p$, and $W_1$ before $W_0$ with probability $1-p$.
When the program coin flip returns 0,
the adversary can win
only with the first linearization, since with the second linearization
it's impossible for $p_2$ to read 0,1.
Similarly, when the program coin flip returns 1,
the adversary can win only with the second linearization.
Thus the contribution of $E$ to the adversary's probability of winning
is at most $(D(E) \cdot p + D(E) \cdot (1-p))/2 = D(E)/2$.

\medskip
\noindent
\textbf{Case 2:}
Consider an execution $E \in {\mathcal E}$ such that both query phases of $R_1$
are finished.

As in the previous case, there are two possible scenarios depending
on whether the random choice for which query phase result to return by $R_1$
is included or not in $E$. In case it is included, the value of $R_1$ is fixed before
the program coin flip and the contribution of $E$ to the adversary's probability of
winning is at most $D(E)/2$. In case it is not, the two query phases either return
the same value which means that the value returned by $R_1$ is again fixed before
the program coin flip, or they return different values. If they return different values,
the ``best'' case for the adversary is that they return 0 and 1 (a $\bot$ value will
make the adversary lose independently of the outcome of the program coin flip).
However, the probability that the value returned by $R_1$ matches the value
returned by the program coin flip in continuations of $E$ is 1/2. Therefore, in both
scenarios, the contribution of $E$
to the adversary's probability of winning is at most $D(E)/2$.

\medskip
\noindent
\textbf{Case 3:}
Consider an execution $E \in {\mathcal E}$ where at least one query
phase of $W_0$ and the first query phase of $R_1$ are pending.

\smallskip
{\em Case 3.1:}
The pending query phase of $W_0$ is its second one.
We say that a query phase {\em sees} a Write if the query phase receives
a reply message with the value and timestamp of that Write.

\smallskip
{\em Case 3.1.1:}
Suppose $W_0$'s first query phase does not see $W_1$.

For all continuations of $E$ in which $W_0$’s update is based on the
first query
(i.e., $W_0$ is linearized before $W_1$) and
the program coin flip is 1,
$p_2$ cannot read 1 followed by 0.
The adversary loses in all such continuations.

Consider continuations of $E$ in which $W_0$’s update is based on its
second query
and the program coin flip is 0.
If this query phase sees $W_1$, then in all these continuations,
$W_1$ is linearized before $W_0$ which implies that
$p_2$ cannot read 0 followed by 1
and the adversary loses.
Therefore, it is in the adversary's interest that the second query
phase of $W_0$ does not see $W_1$,
and thus $W_0$ is linearized before $W_1$.
Then we need to look at $R_1$.
Its second query phase necessarily sees $W_1$ since it starts after $W_1$ finished. Therefore, if $R_1$
returns the value of its second query,
it returns 1, which causes the adversary to lose.
Consequently, at most half of these continuations make the adversary win.

Overall, the contribution of $E$ to the overall win
probability is at most $( D(E)/2 + 3\cdot D(E)/4) /2 = 5 \cdot D(E) / 8$.

\smallskip
{\em Case 3.1.2:}
Suppose $W_0$'s first query phase sees $W_1$.
This case is symmetric to Case 3.1.1, as detailed next.

For all continuations of $E$ in which $W_0$’s update is based on the
first query, i.e., $W_0$ is linearized after $W_1$, and the program coin
flip is 0, $p_2$ cannot read 0 followed by 1.
The adversary loses in all such continuations.
The continuations where $W_0$’s update is based on the
second query admit precisely the same argument as in Case 3.1.1.

The contribution of this execution remains at most $5 \cdot D(E) / 8$.

\smallskip
{\em Case 3.2:}  The pending query phase of $W_0$ is its first one.

Thus $W_0$'s second query phase is guaranteed to see $W_1$.
When the second query phase of $W_0$ is used, i.e.,
$W_0$ is linearized after $W_1$, the read $R_2$ cannot return 1,
implying that if the coin flip is 0, then the adversary loses.
Thus the adversary wins in at most half of these continuations.
The continuations where $W_0$’s update is based on the
first query admit precisely the same argument as the continuations in Case 3.1.1
based on the second query.

The overall contribution remains at most $5 \cdot D(E) / 8$.

\medskip
\noindent
\textbf{Case 4:}
Consider an execution $E \in {\mathcal E}$ where at least one query
phase of $W_0$ and the second query of $R_1$ are pending.
Similar to the previous cases, we can show that
the contribution of $E$ to the adversary's probability of winning
is at most $5 \cdot D(E) / 8$. 

\medskip
Summing over all the cases shows that
the maximum probability of an adversary winning is
$( \sum_{E \in {\mathcal E}} D(E) ) \cdot 5/8 = 5/8$.

\end{document}